\documentclass[10pt, conference]{IEEEtran} 
\IEEEoverridecommandlockouts

\def\BibTeX{{\rm B\kern-.05em{\sc i\kern-.025em b}\kern-.08em
    T\kern-.1667em\lower.7ex\hbox{E}\kern-.125emX}}
\usepackage[utf8]{inputenc}
\usepackage{graphicx}
\usepackage[justification=centering,font=small]{caption}
\usepackage{subcaption}
\usepackage{amsmath, amssymb}
\usepackage{bm}
\usepackage{optidef}
\usepackage{amsthm}

\usepackage{amsmath, amssymb, amsfonts, amsthm}
\usepackage{cancel}
\usepackage{bm}
\usepackage{cite}
\usepackage{graphicx}
\usepackage{algorithm}
\usepackage{algorithmic}
\usepackage{color}
\usepackage{amssymb}
\usepackage{booktabs}
\usepackage{algorithm}
\usepackage{algorithmic}

\newtheorem{theorem}{Theorem}

\allowdisplaybreaks
\newtheorem{lemma}{Lemma} 
\newtheorem{proposition}{Proposition}
\allowdisplaybreaks[4]

\begin{document}
\title{Energy-Efficient Resource Allocation for Six-Dimensional Movable Antenna Systems} 


\author{Ziyun Zhang, Ruotong Zhao, Shaokang Hu, and Derrick Wing Kwan Ng,~\IEEEmembership{Fellow,~IEEE} \\
\thanks{The work of Z. Zhang and D. W. K. Ng was supported by the Australian Research Council's Discovery Projects (DP260100642). The work of R. Zhao was supported in part by the Commonwealth through an Australian Government Research Training Program Scholarship [DOI: https://doi.org/10.82133/C42F-K220] and the Australian Research Council's Discovery Projects (DP240101019). The work of S. Hu was supported by the Australian Research Council's Discovery Projects (DP240101019).}
\IEEEauthorblockA{School of Electrical Engineering and Telecommunications, University of New South Wales, Sydney, NSW 2052, Australia
}}

\maketitle
\begin{abstract}

This paper investigates the energy-efficiency (EE) maximization problem for a multiuser wireless network equipped with six-dimensional movable antennas (6DMAs), where the three-dimensional (3D) positions and  orientations of the antennas are jointly optimized to fully exploit the additional spatial degrees of freedom offered by dynamic channel reconfiguration. However, the practical operation of 6DMAs incurs non-negligible mechanical energy consumption. Moreover, orientation-dependent phase variations, together with the strong coupling among antenna positions, rotation angles, transmit beamforming, and time allocation, render the resulting problem highly non-convex and analytically challenging. To address this issue, we develop a block coordinate descent (BCD) optimization framework that integrates Dinkelbach’s transformation with the majorization-minimization (MM) approach to efficiently obtain a high-quality suboptimal solution with guaranteed convergence. Simulation results unveil that the proposed design achieves significant EE improvements over conventional benchmarks, thereby highlighting the critical importance of accounting for practical mechanical energy costs in 6DMA-enabled systems. Furthermore, our results reveal a fundamental trade-off between throughput enhancement and mechanical overhead: although larger antenna reconfigurations can improve channel conditions, their EE gains gradually diminish due to the increased mechanical energy consumption.

\end{abstract}
\large\normalsize
\section{Introduction}
The upcoming sixth-generation (6G) wireless networks are expected to deliver unprecedented spectral efficiency and energy efficiency (EE) to satisfy rapidly escalating data demands under stringent power constraints~\cite{shi20256dma}. Conventional wireless systems have relied on increasingly large antenna arrays to enhance performance, such as multiple-input multiple-output (MIMO) and massive MIMO architectures. However, continuously increasing the number of antennas inevitably leads to increased hardware cost, energy consumption, and signal-processing complexity. Recently, movable antenna (MA) technology has emerged as a promising approach to efficiently harness additional spatial degrees of freedom (DoF) in wireless channels by dynamically reconfiguring antenna positions to adapt to the propagation environment and enhance channel conditions~\cite{zhang2026movableenergy}. Building upon this concept, the six-dimensional movable antenna (6DMA) architecture has been proposed as an advanced MA paradigm that enables joint adjustment of the three-dimensional (3D) positions and orientations of a limited number of BS antennas, thereby providing greater flexibility in exploiting spatial DoF and improving system performance~\cite{shao_6d_2025-1}.

To fully exploit the potential of 6DMA systems, recent studies have begun to reveal the performance gains of 6DMA-enabled wireless networks in terms of sum-rate maximization under various system settings, including multiuser channel modeling for 6DMA systems~\cite{shao_6dma_2025}, joint optimization of antenna positions and rotation angles under discrete reconfiguration constraints~\cite{shao_6d_2025-1}, and the design of 6DMA-assisted cell-free MIMO networks~\cite{shi20256dma}. Nevertheless, existing works have largely overlooked a critical practical issue, namely,  the mechanical energy consumption incurred by both antenna movement and rotation. Although several works on conventional MA systems have incorporated movement-related power consumption into EE optimization~\cite{zhang2026movableenergy,ding_energy_2025,wei_energy-efficient_2025}, they typically consider only translational motion and neglect the energy cost of antenna rotation. Consequently, EE optimization for 6DMA-enabled networks under a practical mechanical energy consumption model that captures both translational and rotational motions remains largely unexplored. More importantly, ignoring these mechanical energy costs may lead to overly optimistic EE evaluations and impractical resource allocation designs.

Motivated by these research gaps and considerations, this paper studies the EE maximization problem in a multiuser 6DMA system through the joint design of transmit beamforming, time allocation, antenna positions, and antenna orientations, while explicitly accounting for the mechanical energy consumption induced by antenna motion. The resulting problem is highly non-convex due to the strong coupling among the optimization variables, the fractional-form objective function, and the orientation/position-dependent phase variations. To tackle this challenge, we develop a block coordinate descent (BCD)-based algorithm that integrates Dinkelbach’s transformation, majorization-minimization (MM) technique, semidefinite relaxation (SDR), and a penalty-based framework to obtain a high-quality suboptimal solution. Numerical results demonstrate that the proposed 6DMA design significantly outperforms conventional MA systems and fixed-position antenna architectures, highlighting the importance of accounting for energy consumption in the 6DMA system to meticulously balance the trade-off between data rate and mechanical overhead.

\textit{Notation:} Boldface uppercase and lowercase letters denote matrices and vectors, respectively. $\mathbb{R}^N$ and $\mathbb{C}^N$ denote the $N$-dimensional real and complex Euclidean spaces, respectively. $\mathbf{R}(\cdot)$ denotes a $3$D rotation matrix. For any matrix $\mathbf{A}$, $\mathbf{A}^T$, $\mathbf{A}^H$, $\operatorname{Tr}(\mathbf{A})$, $\operatorname{Rank}(\mathbf{A})$, and $[\mathbf{A}]_{i,j}$ denote its transpose, Hermitian transpose, trace, rank, and $(i,j)$-th entry, respectively. $\|\cdot\|_F$, $\|\cdot\|_2$, and $\|\cdot\|_*$ denote the Frobenius, spectral, and nuclear norms, respectively. $\otimes$ denotes the Kronecker product. $\lambda_{\max}(\mathbf{A})$ and $\mathbf{u}_{\max}(\mathbf{A})$ denote the largest eigenvalue of $\mathbf{A}$ and its associated eigenvector, respectively. 
$\mathfrak{Re}\{\cdot\}$ and $\mathfrak{Im}\{\cdot\}$ denote the real and imaginary parts, respectively. $\mathbb{E}\{\cdot\}$ denotes expectation. $\mathcal{CN}(\mu,\sigma^2)$ denotes a circularly symmetric complex Gaussian distribution with mean $\mu$ and variance $\sigma^2$. $\frac{\partial f}{\partial x}$ denotes the partial derivative of $f$ with respect to $x$, and $\arccos(\cdot)$ denotes the inverse cosine function.


\section{System Setup}
\subsection{System Model}
As shown in Fig. \ref{fig:system_model}, we consider a downlink multiuser multiple-input single-output (MISO) system where a base station (BS) serves $K$ single-antenna users with $\mathcal{K} \triangleq \{1,\dots,K\}$. The BS comprises $B$ independently adjustable antenna surfaces, $\mathcal{B} \triangleq \{1, \ldots, B\}$, each with $N_s$ elements and six controllable mechanical DoFs~\cite{shao_6d_2025-1}, yielding a total of $N = B N_s$ antennas. During downlink transmission, the BS sends independent data symbols $s_k \in \mathbb{C}$ to the $k$-th user, with $\mathbb{E}\{|s_k|^2\} = 1, \forall k$.
As illustrated in Fig. \ref{fig:feasible_region}(a), the $b$-th surface is specified by its center $\mathbf{q}_b = [x_b, y_b, z_b]^T \in \mathcal{V}$, while Fig. \ref{fig:feasible_region}(b) depicts its orientation  $\boldsymbol{\theta}_b = [\alpha_b, \beta_b, \gamma_b]^T \in \mathcal{T}$, where $\mathcal{V} = [\mathbf{q}^{\min}, \mathbf{q}^{\max}]$, $\mathcal{T} = [\boldsymbol{\theta}^{\min}, \boldsymbol{\theta}^{\max}]$, and $|\beta_b| \le \beta_{\max} < \pi/2$~\cite{shao_6d_2025-1}.
\begin{figure}[!t]
    \centering
  
    \includegraphics[width=0.98\columnwidth]{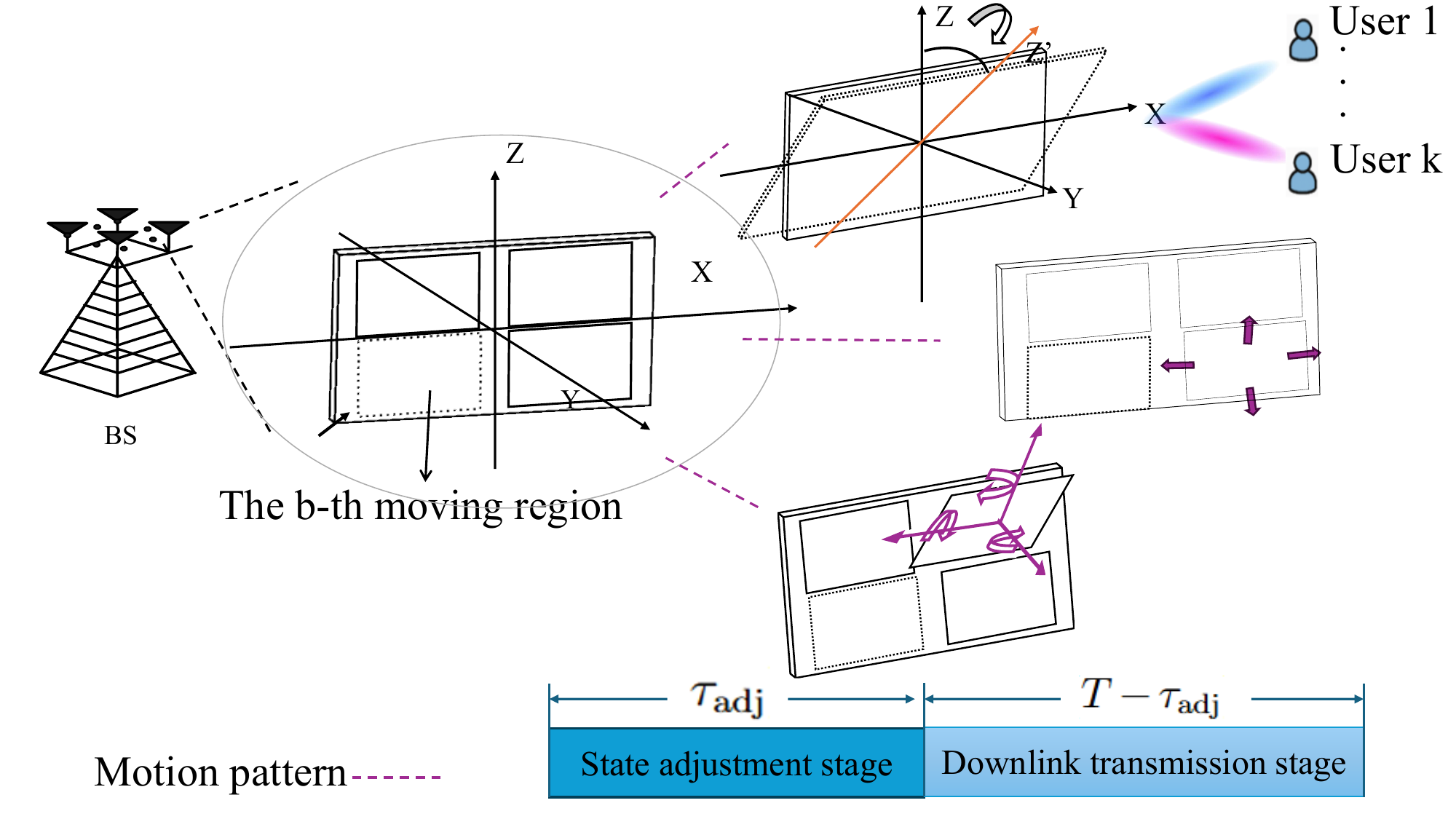} 

    \captionsetup{singlelinecheck=off, justification=raggedright}
    
    \caption{System model of the 6DMA network. }
    \label{fig:system_model}
    \vspace{-4mm}
\end{figure}

Let $\bar{\mathbf{r}}_n \in \mathbb{R}^3$ be the local coordinate of the $n$-th antenna element relative to $\mathbf{q}_b$, its corresponding global position is
\begin{equation}
    \mathbf{r}_{b,n}(\mathbf{q}_b, \boldsymbol{\theta}_b) = \mathbf{q}_b + \mathbf{R}(\boldsymbol{\theta}_b) \bar{\mathbf{r}}_n\in \mathbb{R}^3, \ \forall n \in \mathcal{N}_s, \forall b \in \mathcal{B},
    \label{eq:position_model}
\end{equation}
where $\mathbf{R}(\boldsymbol{\theta}_b) = \mathbf{R}_z(\gamma_b) \mathbf{R}_y(\beta_b) \mathbf{R}_x(\alpha_b) \in \mathbb{R}^{3 \times 3}$ is the rotation matrix for orientation $\boldsymbol{\theta}_b \in \mathbb{R}^3$, with the rotation matrices:
\begin{align}
    \mathbf{R}_x(\alpha_b) &= \scalebox{0.85}{$\begin{bmatrix} 1 & 0 & 0 \\ 0 & \cos\alpha_b & -\sin\alpha_b \\ 0 & \sin\alpha_b & \cos\alpha_b \end{bmatrix}$}, \nonumber \\
    \mathbf{R}_y(\beta_b) &= \scalebox{0.85}{$\begin{bmatrix} \cos\beta_b & 0 & \sin\beta_b \\ 0 & 1 & 0 \\ -\sin\beta_b & 0 & \cos\beta_b \end{bmatrix}$}, \nonumber \\
    \mathbf{R}_z(\gamma_b) &= \scalebox{0.85}{$\begin{bmatrix} \cos\gamma_b & -\sin\gamma_b & 0 \\ \sin\gamma_b & \cos\gamma_b & 0 \\ 0 & 0 & 1 \end{bmatrix}$}.
    \label{3}
\end{align}
Since the local geometry $\{\bar{\mathbf{r}}_n\}$ is fixed, each 6DMA surface is characterized by its position $\mathbf{q}_b$ and orientation $\boldsymbol{\theta}_b$. To avoid physical collisions and excessive mutual coupling among the surfaces, we enforce $\|\mathbf{q}_i - \mathbf{q}_j\|_2 \ge D_{\min},\ \forall i,j \in \mathcal{B},\ i \neq j$~\cite{shao_6d_2025-1}.

\subsection{Energy Consumption Model}
\label{sec:mech_move_power}
Assuming quasi-static block fading~\cite{wei_energy-efficient_2025}, each block of duration $T$ includes a mechanical adjustment phase of duration $\tau_{\mathrm{adj}}$ to configure the 6DMA poses ($\mathbf{q}_b, \boldsymbol{\theta}_b$) and a downlink transmission phase of duration $T - \tau_{\mathrm{adj}}$.
Moreover, we adopt \textit{Euler's rotation theorem}~\cite{spong_robot_2006,murray_mathematical_1994} for describing the rotational energy which avoids the Euclidean metric inconsistencies\footnote{
Specifically, due to the gimbal lock effect~\cite{murray_mathematical_1994}, the Euclidean distance $\|\boldsymbol{\theta}_b - \boldsymbol{\theta}_b^{(0)}\|_2$ may not faithfully reflect the actual mechanical effort of rotation, since a small pitch variation near $90^\circ$ can induce large roll and yaw changes.}~\cite{ding_energy_2025, wei_energy-efficient_2025}. Specifically, the rotation angle $\Delta \phi_b \in [0, \pi]$ is determined by $\mathbf{R}_{\Delta,b} = \mathbf{R}(\boldsymbol{\theta}_b) \mathbf{R}^T(\boldsymbol{\theta}_b^{(0)})$ with 
$\Delta \phi_b = \arccos\big( \tfrac{\text{Tr}(\mathbf{R}_{\Delta,b}) - 1}{2} \big)$, where the corresponding rotation axis\footnote{Two edge cases deserve attention: 1) when $\Delta \phi_b = 0$, the required rotation energy is zero and the rotation axis $\mathbf{n}_b$ can be chosen arbitrarily; 2) when $\Delta \phi_b = \pi$, the expression becomes singular, and $\mathbf{n}_b$ is given by the unit eigenvector of $\mathbf{R}_{\Delta,b}$ associated with eigenvalue $1$.} is~\cite{shi20256dma,shao_6d_2025-1}:
\begin{equation}
     \mathbf{n}_b(\Delta\phi_b)= \frac{1}{2\sin(\Delta \phi_b)} \begin{bmatrix} [\mathbf{R}_{\Delta,b}]_{3,2} - [\mathbf{R}_{\Delta,b}]_{2,3} \\ [\mathbf{R}_{\Delta,b}]_{1,3} - [\mathbf{R}_{\Delta,b}]_{3,1} \\ [\mathbf{R}_{\Delta,b}]_{2,1} - [\mathbf{R}_{\Delta,b}]_{1,2} \end{bmatrix} \in \mathbb{R}^3.
\end{equation}
Assuming translational and angular speeds $v_b$ and $\omega_b$, the translation and rotation durations of surface $b$ are given by $\tau_b^{\text{mov}} = \frac{\|\mathbf{q}_b - \mathbf{q}_b^{(0)}\|_2}{v_b}$ and $\tau_b^{\text{rot}} = \frac{\Delta\phi_b}{\omega_b}$. Accordingly, the adjustment delay is $\tau_\mathrm{adj} = \max_{b \in \mathcal{B}} \{ \max ( \tau_b^{\text{mov}}, \tau_b^{\text{rot}} ) \} \le T$.

\begin{figure}[t]
    \centering
    \captionsetup{singlelinecheck=off, justification=raggedright}

    \begin{subfigure}[t]{0.5\linewidth}
        \centering
        \includegraphics[width=\linewidth]{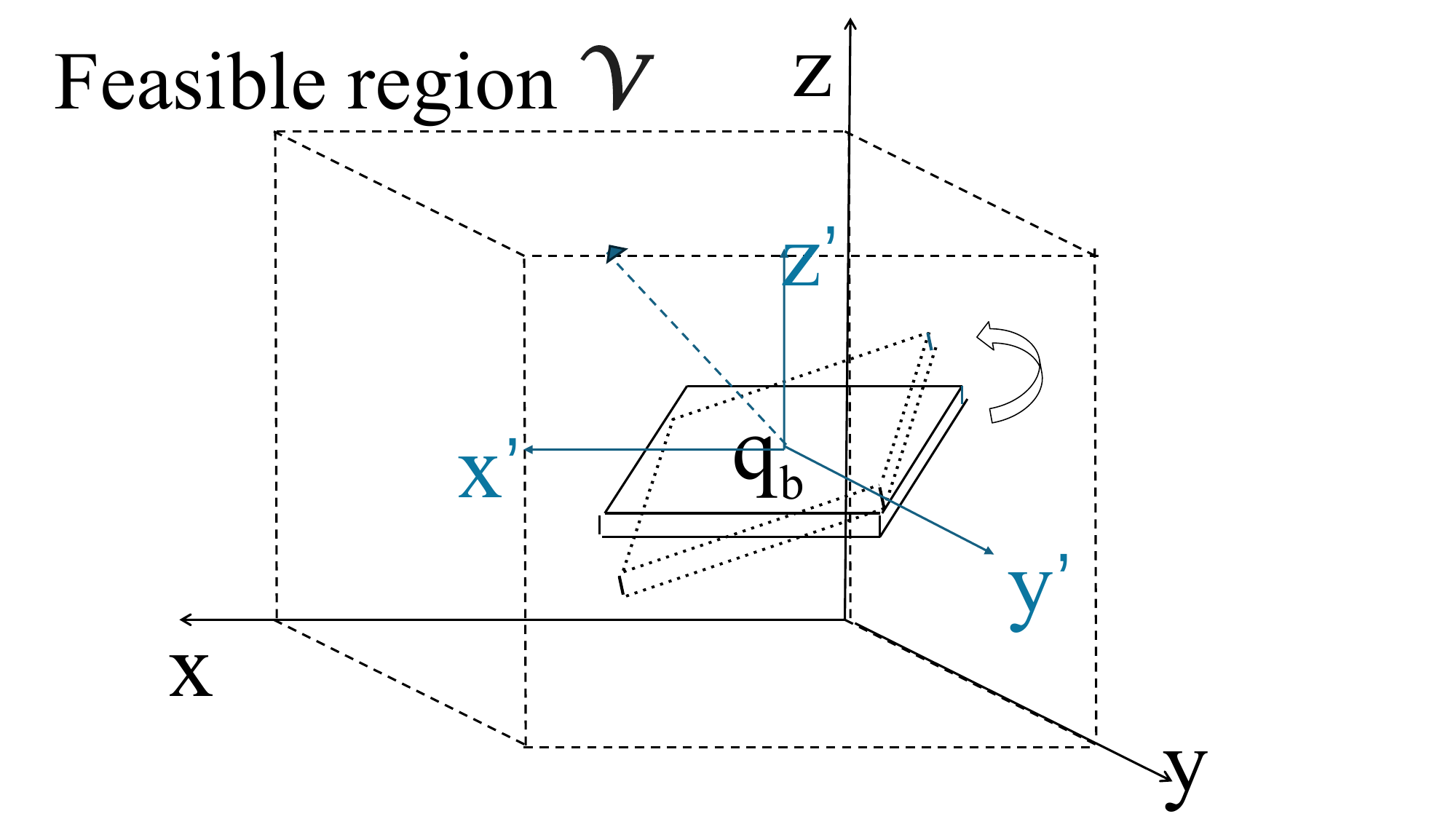}
        \caption{Feasible region $\mathcal{V}$ and center $\mathbf{q}_b$.}
        \label{fig:region}
    \end{subfigure}\hfill
    \begin{subfigure}[t]{0.5\linewidth}
        \centering
        \includegraphics[width=\linewidth]{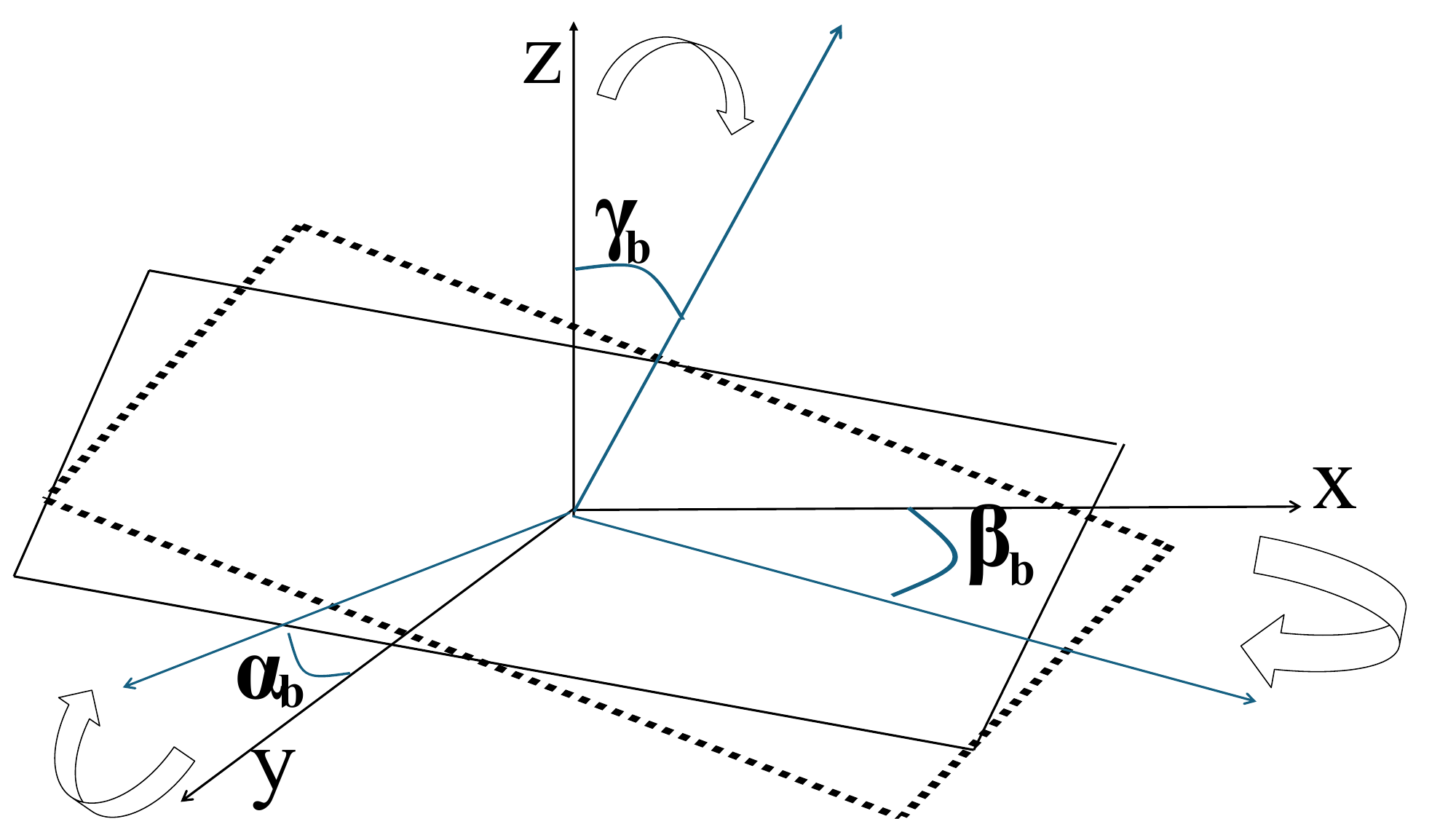}
        \caption{Local frame of the $b$-th surface.}
        \label{fig:rotation}
    \end{subfigure}

    \caption{Geometry of the $b$-th 6DMA surface.}
    \label{fig:feasible_region}

\end{figure}

Define the mechanical energy as $E_{\mathrm{mech}} \triangleq \sum_{b=1}^{B}(E_{\mathrm{move},b}+E_{\mathrm{rotate},b})$. Then, the energy consumption over block $T$ is:
\begin{equation}
    E_{\text{total}} = E_{\text{mech}} + (T - \tau_{\mathrm{adj}}) P_{\text{tx}} + T P_s,
    \label{eq:E_total}
\end{equation}
where $P_{\text{tx}}$ and $P_s$ denote the transmit and static circuit powers, respectively~\cite{zhang2026movableenergy}. 
By neglecting acceleration effects~\cite{wei_energy-efficient_2025}, the translational power consumption is given by $P_{\text{move}}(v)=M_{\text{step}}\frac{v}{l_0}$, where $l_0$ is the actuator step length and $M_{\text{step}}$ is the translational actuator torque~\cite{ding_energy_2025}. The movement energy consumed by surface $b$ is $E_{\text{move},b} = \tau_b^{\text{mov}} P_{\text{move}}(v_b)$ while the rotation energy consumed by surface $b$ is $E_{\text{rotate},b} = \tau_b^{\text{rot}} M_{\text{rot}}\omega_b$~\cite{murray_mathematical_1994}, where $M_{\text{rot}}$ denotes the rotational actuator torque. 
Moreover, for a given transmission duration and beamforming matrix $\mathbf{W} = [\mathbf{w}_1, \dots, \mathbf{w}_K] \in \mathbb{C}^{N \times K}$, where $\mathbf{w}_k \in \mathbb{C}^{N \times 1}$ denotes the beamforming vector for the $k$-th user, the transmit power is $P_{\text{tx}}(\mathbf{W}) = \text{Tr}(\mathbf{W}\mathbf{W}^H)$. 

\subsection{Performance Metric}
This work adopts a plane-wave propagation model under line-of-sight (LoS)-dominated far-field conditions\cite{ shao_6d_2025-1, shao_6dma_2025}. The wireless downlink channel from the 6DMA-enabled BS to the $k$-th user is given by
$\mathbf{h}_k(\mathbf{q},\boldsymbol{\theta}) = [
    \mathbf{h}_{k,1}^T(\mathbf{q}_1, \boldsymbol{\theta}_1), \ldots, 
    \mathbf{h}_{k,B}^T(\mathbf{q}_B, \boldsymbol{\theta}_B)]^T \in \mathbb{C}^{N \times 1}, \forall k$, where $\mathbf{q} = [\mathbf{q}_1^T, \dots, \mathbf{q}_B^T]^T \in \mathbb{R}^{3B \times 1}$ and $\boldsymbol{\theta} = [\boldsymbol{\theta}_1^T, \ldots, \boldsymbol{\theta}_B^T]^T \in \mathbb{R}^{3B \times 1}$. In particular, the subchannel corresponding to the $b$-th surface is given by $\mathbf{h}_{k,b}(\mathbf{q}_b, \boldsymbol{\theta}_b) = g_k \mathbf{a}_{k,b}(\mathbf{q}_b, \boldsymbol{\theta}_b, \varphi_k, \vartheta_k) \in \mathbb{C}^{N_s \times 1}, \forall k, b$~\cite{shao_6d_2025-1}, where $g_k \in \mathbb{C}$ and $\mathbf{a}_{k,b} \in \mathbb{C}^{N_s \times 1}$ denote the complex gain and the array response, respectively~\cite{shao_6d_2025-1}. The $n$-th element of $\mathbf{a}_{k,b}$ is
$[\mathbf{a}_{k,b}]_n = \sqrt{G_0(\varphi_k, \vartheta_k, \boldsymbol{\theta}_b)} e^{-j \Psi_{b,n,k}}$,
where the direction vector is $\mathbf{u}(\varphi_k, \vartheta_k) = [\cos\vartheta_k\cos\varphi_k, \cos\vartheta_k\sin\varphi_k, \sin\vartheta_k]^T$ and $\Psi_{b,n,k} = \frac{2\pi}{\lambda} \mathbf{u}^T(\varphi_k, \vartheta_k) \mathbf{r}_{b,n}$ with wavelength $\lambda$ and orientation-aware element gain $G_0(\cdot)$~\cite{shao_6d_2025-1, 7913628}:
\begin{equation}
G_0(\varphi_k, \vartheta_k, \bm{\theta}_b) = 
\begin{cases}
G \left( \mathbf{u}_k^T\mathbf{z}_b \right)^\eta, & \text{if } \mathbf{u}_k^T\mathbf{z}_b > 0, \\
0, & \text{otherwise}.
\end{cases}
\label{eq:antenna_gain}
\end{equation}
where $\eta \ge 0$ is the beamwidth factor and $G = 2(\eta + 1)$  enforces energy conservation. The rotated normal vector is $\mathbf{z}_b(\bm{\theta}_b) = \mathbf{R}(\bm{\theta}_b)[0, 0, 1]^T$, such that $\mathbf{u}_k^T \mathbf{z}_b$ represents the cosine of the incident angle relative to the boresight.

Hence, the received signal at the $k$-th user is $y_{k} = \mathbf{h}_{k}^{H}\mathbf{w}_{k}s_{k} + \sum_{i\neq k}\mathbf{h}_{k}^{H}\mathbf{w}_{i}s_{i} + n_{k}$, where $n_k \sim \mathcal{CN}(0, \sigma_k^2)$ denotes the additive white Gaussian noise (AWGN), and $\sigma_k^2$ is the corresponding noise power. The resulting signal-to-interference-plus-noise ratio (SINR) is:
\begin{equation}
    \gamma_k(\mathbf{W}, \mathbf{q}, \boldsymbol{\theta}) = \frac{|\mathbf{h}_k^H(\mathbf{q}, \boldsymbol{\theta})\mathbf{w}_k|^2}{\sum_{i \neq k} |\mathbf{h}_k^H(\mathbf{q}, \boldsymbol{\theta})\mathbf{w}_i|^2 + \sigma_k^2}, \quad \forall k,
    \label{eq:sinr}
\end{equation}
and the corresponding achievable data rate and EE are $R_k = \log_2(1 + \gamma_k)$ and  $\text{EE} = \tfrac{(T-\tau_{\mathrm{adj}})\sum_{k=1}^{K} R_k}{E_{\text{total}}}$, respectively.
    
\section{Problem Formulation}
We jointly optimize $\mathbf{W}$, $\mathbf{q}$, $\boldsymbol{\theta}$, and $\tau_{\text{adj}}$ to maximize the system EE, which is formulated as:
\begin{flalign}
\label{prob:P1}
\begin{aligned}
    & \max_{\mathbf{W}, \mathbf{q}, \boldsymbol{\theta}, \tau_{\mathrm{adj}}} \quad \text{EE} \\
    & \text{s.t.} \; \text{C1}: \text{Tr}(\mathbf{W}\mathbf{W}^H) \le P_{\max}, \\
    & \quad\;\;  \text{C2}: \tfrac{T - \tau_{\mathrm{adj}}}{T} R_k \ge R_{\min}, \, \forall k, \\
    & \quad\;\; \text{C3}: \tfrac{\|\mathbf{q}_b - \mathbf{q}_b^{(0)}\|_2}{v_b} \le \tau_{\mathrm{adj}}, \, \forall b, \; \text{C4}: \tfrac{\Delta\phi_b}{\omega_b} \le \tau_{\mathrm{adj}}, \, \forall b, \\
    & \quad\;\; \text{C5}: \|\mathbf{q}_i - \mathbf{q}_j\|_2^2 \ge D_{\min}^2, \quad \forall i, j \in \mathcal{B}, i \neq j, \\
    & \quad\;\; \text{C6}: \mathbf{q}_b \in \mathcal{V}, \boldsymbol{\theta}_b \in \mathcal{T}, \, \forall b, \quad \text{C7}: 0 \le \tau_{\mathrm{adj}} \le T.
\end{aligned} &&
\end{flalign}
where constraints C1--C2 ensure the transmit power budget, $P_{\max}$ and the quality of service (QoS) requirements, $R_{\min}$, respectively. Constraints C3--C5 impose the mechanical, delay, and collision-avoidance constraints, respectively, while C6 and C7 specify the feasible position and orientation ranges and block-duration limit, respectively.

\section{Optimization Solution}
Since problem \eqref{prob:P1} is highly non-convex due to the coupling among the optimization variables and the fractional-form objective, we apply semidefinite lifting by defining$\mathbf{W}_k=\mathbf{w}_k\mathbf{w}_k^H$ and $\mathbf{H}_k=\mathbf{h}_k\mathbf{h}_k^H$, so that $|\textbf{h}_k^H\textbf{w}_j|^2 = \text{Tr}(\textbf{H}_k\textbf{W}_j), \forall k,j \in \mathcal{K}$. We then employ a BCD-based framework to alternately optimize four variable blocks: beamforming matrices $\mathbf{W}$, antenna positions $\mathbf{q}$, antenna orientations $\boldsymbol{\theta}$, and time allocation $\tau_{\text{adj}}$, thereby obtaining a high-quality suboptimal solution.

\subsection{Block 1: Optimization of Beamforming Matrices}
We optimize the beamforming matrices $\mathbf{W}$ for given $\mathbf{q}$, $\boldsymbol{\theta}$, and $\tau_{\text{adj}}$. The resulting problem can be written as:
\begin{equation} \label{eq: formul_1}
\begin{aligned}
    & \underset{\{\mathbf{W}\}}{\text{maximize}} \quad \text{EE} \\
    \text{s.t.} \; 
    & \text{C1, C2, C8}\hspace{-1mm}: \mathbf{W}_k \succeq \mathbf{0}, \forall k, \; \text{C9}: \text{Rank}(\mathbf{W}_k) \le 1, \forall k.
\end{aligned}
\end{equation}
Since \eqref{eq: formul_1} remains non-convex, we recast the rate expression to a difference-of-convex (DC) form, i.e., $R_k = f_{1,k} - f_{2,k}, \forall k$, where $f_{1,k} \triangleq \log_2 ( \sum_{i \in \mathcal{K}} \text{Tr}(\mathbf{H}_k \mathbf{W}_i) + \sigma_k^2 )$ and $f_{2,k} \triangleq \log_2 ( \sum_{j \neq k} \text{Tr}(\mathbf{H}_k \mathbf{W}_j) + \sigma_k^2 )$. Since $f_{2,k}$ is concave, we apply the MM method and utilize the first-order Taylor expansion to obtain the affine lower bound $R_k(\mathbf{W}) \ge \tilde{R}_k^{(t_1)}(\mathbf{W})$, as
\begin{equation}
\label{eq:R_k_surrogate}
\begin{aligned}
    \tilde{R}_k^{(t_1)} &= f_{1,k} - f_{2,k}^{(t_1)} \\
    &\quad - \sum_{i \neq k} \frac{\text{Tr}(\mathbf{H}_k(\mathbf{W}_i - \mathbf{W}_i^{(t_1)}))}{(\ln 2) \left( \sum_{j \neq k} \text{Tr}(\mathbf{H}_k \mathbf{W}_j^{(t_1)}) + \sigma_k^2 \right)}, \forall k,
\end{aligned}
\end{equation}
where ($t_1$) denotes the solution of the $t_1$-th MM iteration.
Hence, applying Dinkelbach's method in Appendix A, a suboptimal solution to~\eqref{eq: formul_1} can be obtained by solving the $(t_1+1)$-th MM iteration with a given parameter $\lambda^{(m)}$, as:
\begin{equation} \label{eq:formul_1.1}
\begin{aligned}
    \max_{\{\mathbf{W}_k\}} \quad & (T - \tau_{\text{adj}}) \sum_{k} \big( \tilde{R}_{k}^{(t_1)} - \lambda^{(m)} \text{Tr}(\mathbf{W}_k) \big) \\
    &\quad - \lambda^{(m)} (E_{\text{mech}} + T P_s) \\
    \text{s.t.} \quad & \text{C1}, \ \overline{\text{C2}}: \frac{T - \tau_{\text{adj}}}{T}\tilde{R}_{k}^{(t_1)} \ge R_{\min}, \ \text{C8}.
\end{aligned}
\end{equation}
where $\overline{\text{C2}} \Rightarrow \text{C2}$ and the rank constraint C9 is relaxed such that problem~\eqref{eq:formul_1.1} is convex and can be solved by a standard convex programming solver~\cite{hu_sum-rate_2025}. The tightness of the SDR is revealed in the following theorem. 

\begin{theorem} \label{thm:Rank_one}
If $P_{\max},R_{{\rm min}} > 0$, and~\eqref{eq:formul_1.1} is feasible, then a solution satisfying $\text{Rank}(\textbf{W}_k) \le 1$, can always be constructed.
\end{theorem}

\begin{proof}
Due to space limitations, only a proof sketch is provided. By examining the Karush-Kuhn-Tucker conditions of \eqref{eq:formul_1.1}, there always exists an optimal rank-one solution $\textbf{W}_k^\star$. Moreover, $\textbf{W}_k^\star$ can be explicitly constructed from the optimal dual variables of the corresponding dual problem.
\end{proof}

\subsection{Block 2: Optimization of Antenna Positions}
In this subsection, we optimize the antenna position vectors $\mathbf{q}$ with given $\mathbf{W}$, $\boldsymbol{\theta}$, and $\tau_{\text{adj}}$, which can be formulated as:
\begin{align}
\label{eq: formul_2}
&\underset{\{\mathbf{q}\}}{\text{maximize}} \quad\text{EE} \nonumber \\
\text{s.t.} & \ \text{C2}, \text{C3}, \text{C5}, \text{C6}.
\end{align}
For exposition, we absorb the path-loss and array response terms into $\bar{\mathbf{h}}_{k,b} \in \mathbb{C}^{N_s \times 1}$, and denote the resulting channel model by $\mathbf{h}_{k,b}(\mathbf{q}_b)=\bar{\mathbf{h}}_{k,b}e^{-j\frac{2\pi}{\lambda}\mathbf{u}_k^T\mathbf{q}_b}$.
With the phase-shift vector $\mathbf{a}_k(\mathbf{q})=\bigl[e^{-j\frac{2\pi}{\lambda}\mathbf{u}_k^T\mathbf{q}_1},\dots,e^{-j\frac{2\pi}{\lambda}\mathbf{u}_k^T\mathbf{q}_B}\bigr]^T$, we introduce $\mathbf{X}_k=\mathbf{a}_k(\mathbf{q})\mathbf{a}_k^H(\mathbf{q})$ such that $\operatorname{Tr}(\mathbf{H}_k\mathbf{W}_i)\!=\!\operatorname{Tr}(\mathbf{M}_{k,i}\mathbf{X}_k)$, where $[\mathbf{M}_{k,i}]_{b,b'}=\bar{\mathbf{h}}_{k,b}^H\mathbf{W}_i^{b,b'}\bar{\mathbf{h}}_{k,b'}$, $\forall\, b,b',k$. Consequently, problem~\eqref{eq: formul_2} can be reformulated as:
\begin{equation} \label{eq: formul_2.1}
\begin{aligned}
    & \underset{\mathbf{q},\{\mathbf{X}_k\}}{\text{maximize}} \quad \text{EE} \\
    \text{s.t.} \quad 
    & \text{C2}, \text{C3}, \text{C5}, \text{C6}, \, \text{C10}: \mathbf{X}_k \succeq \mathbf{0}, \, \forall k, \\
    & \text{C11}: \text{Rank}(\mathbf{X}_k) = 1, \, \text{C12}: \text{diag}(\mathbf{X}_k) = \mathbf{1}, \, \forall k, \\
    & \text{C13}: [\mathbf{X}_k]_{b,b'} = e^{-j \Delta \psi_{k,b,b'}},
\end{aligned}
\end{equation}
where $\Delta \psi_{k,b,b'} \triangleq \frac{2\pi}{\lambda} \mathbf{u}_k^T (\mathbf{q}_b - \mathbf{q}_{b'})$ for notational simplicity. Moreover, to effectively handle the rank-one constraint C11, we introduce the following lemma \cite{conference}:
\begin{lemma} \label{lem:Rank_one}
The rank-one constraint C11 is equivalent to constraint $\overline{\mathrm{C11}}:\left\| \mathbf{X}_k \right\|_* - \left\| \mathbf{X}_k \right\|_2 \leq 0, \forall k$.\end{lemma}
\begin{proof}
For any $\textbf{X} \in \mathbb{H}^n \succeq \textbf{0}$ , the inequality $\left\|\textbf{X}\right\|_* = \sum_i\Lambda_i \geq \left\|\textbf{X}\right\|_2 = \text{max}\{\Lambda_i\}$ always holds, where $\Lambda_i$ is the $i$-th singular value of $\textbf{X}$. The equality holds if and only if $\textbf{X}$ is rank-one.
\end{proof}
Although $\overline{\text{C11}}$ remains non-convex, it admits a DC form. We then decouple C13 by applying Euler's formula~\cite{conference, blockage}:
\begin{equation}
\begin{aligned}\label{eq: Euler}
    \text{C13a}: \quad & \mathfrak{Re}\{[\mathbf{X}_k]_{b,b'}\} = \cos(\Delta \psi_{k,b,b'}), && \forall k, b, b', \\
    \text{C13b}: \quad & \mathfrak{Im}\{[\mathbf{X}_k]_{b,b'}\} = -\sin(\Delta \psi_{k,b,b'}), && \forall k, b, b'.
\end{aligned}
\end{equation}
To further handle the challenges caused by the periodic nature of the trigonometric functions, we adopt a penalty-based approach~\cite{yu_irs-assisted_2021,blockage}, in which a phase-consistency violation term is introduced into the numerator of the objective function to preserve the fractional EE structure, yielding
\begin{equation}\label{eq: formul_2.2}
\begin{aligned}
\underset{\mathbf{q}, \{\mathbf{X}_k\}}{\text{maximize}} \quad & \frac{(T - \tau_{\text{adj}}) \sum_{k=1}^{K} R_{k}(\mathbf{X}_k) - \rho \mathcal{P}( \mathbf{X}, \mathbf{q})}{E_{\text{total}}}   \\
\text{s.t.} \quad & \text{C2}, \text{C3}, \text{C5}, \text{C6},
\text{C10},
\overline{\text{C11}},
\text{C12},
\end{aligned}
\end{equation}
where $\mathcal{P}(\mathbf{X}, \mathbf{q}) = \sum_{k, b, b'} \big( | \mathfrak{Re}\{[\mathbf{X}_k]_{b,b'}\} - \cos(\Delta \psi_{k,b,b'}) |^2 + | \mathfrak{Im}\{[\mathbf{X}_k]_{b,b'}\} + \sin(\Delta \psi_{k,b,b'}) |^2 \big)$.

\begin{proposition}
As the penalty factor $\rho \rightarrow \infty$, the penalized problem in \eqref{eq: formul_2.2} becomes equivalent to problem~\eqref{eq: formul_2.1}~\cite{yu_irs-assisted_2021,conference}.
\end{proposition}

Following the same procedure as in~\eqref{eq:R_k_surrogate}, we apply the MM method together with a first-order Taylor approximation at $\mathbf{X}^{(t_2)}$ to problem~\eqref{eq: formul_2.2}, yielding the concave surrogate $\tilde{R}_k^{(t_2)}$. The resulting convex approximation of C2 is denoted by $\widetilde{\mathrm{C2}}$, with the details omitted for brevity. Likewise, the convex approximation of C5 and $\overline{\mathrm{C11}}$ are given by
\begin{align}
    \overline{\text{C5}}: \,\, & 2(\Delta\mathbf{q}_{ij}^{(t_2)})^T \Delta\mathbf{q}_{ij} - \|\Delta\mathbf{q}_{ij}^{(t_2)}\|^2 \geq D_{\min}^2,  \nonumber \\
    \overline{\overline{\text{C11}}}: \,\, & \|\mathbf{X}_k\|_* - \text{Tr}(\mathbf{V}_k^{(t_2)} \mathbf{X}_k) \leq 0, \quad \forall k, \label{eq:C11_linear}
\end{align}
respectively, where $\overline{\text{C5}} \Rightarrow \text{C5}$, $\overline{\overline{\text{C11}}} \Rightarrow  \overline{\text{C11}}$, $\Delta\mathbf{q}_{ij} = \mathbf{q}_i - \mathbf{q}_j$, $\mathbf{V}_k^{(t_2)} = \mathbf{v}_k^{(t_2)} (\mathbf{v}_k^{(t_2)})^H$, and $\mathbf{v}_k^{(t_2)}$ is the principal eigenvector of $\mathbf{X}_k^{(t_2)}$. 
To tackle the non-convex trigonometric penalty term, we employ an MM-based Lipschitz surrogate technique to establish a valid global upper bound~\cite{conference,mairal_incremental_2015}, as
\begin{align}\label{eq:penalty_surrogate}
&\tilde{\mathcal{P}}^{(t_2)}(\mathbf{q,X}) = \mathcal{P}(\mathbf{q}^{(t_2)}, \mathbf{X}) + \nabla_{\mathbf{q}} \mathcal{P}(\mathbf{q}^{(t_2)}, \mathbf{X})^T (\mathbf{q} - \mathbf{q}^{(t_2)}) \nonumber \\
& \hspace{45mm} + \frac{L_P}{2} \|\mathbf{q} - \mathbf{q}^{(t_2)}\|_2^2,
\end{align}
where the Lipschitz constant is given by $L_P = 8\pi^2 KB/\lambda^2$, as derived in Appendix B. By adopting the Dinkelbach's method in Appendix A, the $(t_2 + 1)$-th iteration of the penalty-based MM method for this subproblem can be expressed as: 
\begin{equation}
\label{eq:forum3}
\begin{aligned}
    \max_{\mathbf{q}, \{\mathbf{X}_k\}} \quad & (T - \tau_{\mathrm{adj}}) \sum_k \tilde{R}_k^{(t_2)}(\mathbf{X}_k) \\
    &\quad - \rho \tilde{\mathcal{P}}^{(t_2)}(\mathbf{q}, \mathbf{X}) - \lambda^{(m)} E_{\mathrm{total}} \\
    \text{s.t.} \quad & \text{C1}, \widetilde{\text{C2}}, \text{C3}, \overline{\text{C5}}, \text{C6}, \text{C10}, \overline{\overline{\text{C11}}}, \text{C12}.
\end{aligned}
\end{equation}
where \eqref{eq:forum3} is convex and yields a suboptimal solution to~\eqref{eq: formul_2}.

\subsection{Block 3: Optimization of Antenna Orientations}
Now, we address the optimization of the orientation vector $\boldsymbol{\theta}$ with fixed $\mathbf{W}, \mathbf{q}$, and $\tau_{\text{adj}}$, which is formulated as:
\begin{equation}
\begin{aligned}\label{eq: formul_3}
\quad & \underset{\boldsymbol{\theta}}{\text{maximize}} \quad\text{EE}
\quad \text{s.t.} \quad \text{C2}, \text{C4}, \text{C6}, \text{C7}.
\end{aligned}
\end{equation}
Following Section IV-B, we decompose the channel vector into $\mathbf{h}_k = g_k \mathbf{a}_k(\boldsymbol{\theta})$. By defining the slack matrix $\mathbf{A}_k \triangleq \mathbf{a}_k(\boldsymbol{\theta})\mathbf{a}_k^H(\boldsymbol{\theta})$, the term $|\mathbf{h}_k^H \mathbf{w}_i|^2 = |g_k|^2 \mathrm{Tr}(\mathbf{A}_k \mathbf{W}_i)$ captures the coherent combining gains, allowing $R_k$ to be expressed as a function of $\mathbf{A}_k$.
Hence, the optimization problem in \eqref{eq: formul_3} can be equivalently reformulated into the following subproblem:
\begin{equation}\label{eq:formul_3.1}
\begin{aligned}
    & \max_{\boldsymbol{\theta}, \{\mathbf{A}_{k}\}} \quad \text{EE} \\
    & \text{s.t.} \quad \text{C2}, \text{C4}, \text{C6}, \\
    & \qquad\; \text{C14}: \mathbf{A}_{k} \succeq \mathbf{0}, \, \forall k, \quad \text{C15}: \text{Rank}(\mathbf{A}_{k}) = 1, \, \forall k, \\
    & \qquad\; \text{C16}: [\mathbf{A}_{k}]_{\ell(b,m),\ell(b,n)} = G_{k,b}(\boldsymbol{\theta}_b) e^{-j \Delta \psi_{k,b,m,n}}, \\
    & \qquad\qquad\quad\; \forall k,b,m,n, \\
    & \qquad\; \text{C17}: [\mathbf{A}_{k}]_{\ell(b,n_0),\ell(1,n_0)} = \sqrt{G_{k,b}(\boldsymbol{\theta}_b) G_{k,1}(\boldsymbol{\theta}_1)} \\
    & \qquad\qquad\quad\; \times \exp \big\{ -j \Delta \chi_{k,b} \big\}, \quad \forall k, b\neq1.
\end{aligned}
\end{equation}
where $G_{k,b}(\boldsymbol{\theta}_b)$ and $\mathbf{R}(\boldsymbol{\theta}_b)$ denote the element gain and rotation matrix, respectively. To extract the orientation-dependent responses from the lifted matrix $\mathbf{A}_k$, constraints C16 and C17 characterize the spatial mappings. Specifically, for antenna indices $m, n \in \{1, \dots, N_s\}$, C16 describes the intra-surface phase shift $\Delta \psi_{k,b,m,n} \triangleq (2\pi/\lambda) \mathbf{u}_k^T \mathbf{R}(\boldsymbol{\theta}_b) (\bar{\mathbf{r}}_m - \bar{\mathbf{r}}_n)$. Concurrently, C17 characterizes the inter-surface relative phase difference by taking the first surface as the reference position, yielding $\Delta \chi_{k,b} \triangleq (2\pi/\lambda) \mathbf{u}_k^T [ (\mathbf{q}_b + \mathbf{R}(\boldsymbol{\theta}_b)\bar{\mathbf{r}}_{n_0}) - (\mathbf{q}_1 + \mathbf{R}(\boldsymbol{\theta}_1)\bar{\mathbf{r}}_{n_0}) ]$.

Under the shortest-path rotation assumption, the rotation angle satisfies $\Delta\phi_b = \omega_b \tau_{\text{adj}} \le \pi$. Substituting this into C4 and exploiting the monotonicity of $\cos(\cdot)$ over $[0,\pi]$ yields $\text{Tr}(\mathbf{R}(\boldsymbol{\theta}_b)\mathbf{R}^T(\boldsymbol{\theta}_b^{(0)})) \ge 1 + 2\cos(\omega_b \tau_{\text{adj}}), \forall b$. Since the trace term remains non-convex in $\boldsymbol{\theta}_b$, we employ the MM method with a Lipschitz surrogate, with constant $L_{F,b}$ derived in Appendix B, yielding the following convex subset of C4:
\begin{equation}
\begin{aligned}
    \widehat{\text{C4}}:~& \mathcal{F}(\boldsymbol{\theta}_b^{(t_3)}) + \nabla_{\boldsymbol{\theta}_b} \mathcal{F}(\boldsymbol{\theta}_b^{(t_3)})^T (\boldsymbol{\theta}_b - \boldsymbol{\theta}_b^{(t_3)}) \\
    & \hspace{-2em} - \frac{L_{F,b}}{2} \|\boldsymbol{\theta}_b - \boldsymbol{\theta}_b^{(t_3)}\|_2^2 \ge 1 + 2\cos(\omega_b \tau_{\text{adj}}), \, \forall b \in \mathcal{B}.
\end{aligned}
\end{equation}

Moreover, by \textbf{Lemma \ref{lem:Rank_one}}, C15 admits the equivalent DC representation $\overline{\rm C15}\hspace{-1mm}:\hspace{-1mm} \|\mathbf{A}_{k}\|_* - \|\mathbf{A}_{k}\|_2 \le 0, \forall k$. By applying the decomposition in~\eqref{eq: Euler}, C16 and C17 can be reformulated as:
\begin{equation}
\label{eq:real_imag_decomp}
\resizebox{0.89\hsize}{!}{$
\begin{aligned}
    \text{C16a:} \; & \Re\{[\mathbf{A}_{k}]_{\ell(b,m),\ell(b,n)}\} = G_{k,b} \cos(\Delta \psi_{k,b,m,n}), \\
    \text{C16b:} \; & \Im\{[\mathbf{A}_{k}]_{\ell(b,m),\ell(b,n)}\} = -G_{k,b} \sin(\Delta \psi_{k,b,m,n}), \\
    \text{C17a:} \; & \Re\{[\mathbf{A}_{k}]_{\ell(b,n_0),\ell(1,n_0)}\} = \sqrt{G_{k,b}G_{k,1}} \cos(\Delta \chi_{k,b}), \\
    \text{C17b:} \; & \Im\{[\mathbf{A}_{k}]_{\ell(b,n_0),\ell(1,n_0)}\} = -\sqrt{G_{k,b}G_{k,1}} \sin(\Delta \chi_{k,b}).
\end{aligned}
$}
\end{equation}

Hence, the penalty form of problem~\eqref{eq:formul_3.1} is:
\begin{equation}\label{eq:formul_3.2}
\begin{aligned}
\underset{\boldsymbol{\theta}, \{\mathbf{A}_{k}\}}{\text{maximize}} \quad &\frac{(T - \tau_{\text{adj}}) \sum_{k=1}^{K} R_{k}(\mathbf{A})- \rho_{\theta} \mathcal{P}_{\theta}(\boldsymbol{\theta}, \mathbf{A}) }{E_{\text{total}}} \\
\text{s.t.} \quad & \widehat{\text{C4}}, \text{C6}, \text{C14}, \ \overline{\text{C15}}, 
\end{aligned}
\end{equation}
where 
\begin{equation*}
\begin{aligned}
    & \quad \mathcal{P}_{\theta}(\boldsymbol{\theta}, \mathbf{A}) \\
    &= \sum_{k, b, m \neq n} \Big( \big| \Re\{[\mathbf{A}_{k}]_{\ell(b,m),\ell(b,n)}\} - G_{k,b}\cos(\Delta \psi_{k,b,m,n}) \big|^2 \\
    & + \big| \Im\{[\mathbf{A}_{k}]_{\ell(b,m),\ell(b,n)}\} + G_{k,b}\sin(\Delta \psi_{k,b,m,n}) \big|^2 \Big) \\
    & + \sum_{k, b \ge 2} \Big( \big| \Re\{[\mathbf{A}_{k}]_{\ell(b,n_0),\ell(1,n_0)}\} - \sqrt{G_{k,b} G_{k,1}}\cos(\Delta \chi_{k,b}) \big|^2 \\
    & + \big| \Im\{[\mathbf{A}_{k}]_{\ell(b,n_0),\ell(1,n_0)}\} + \sqrt{G_{k,b} G_{k,1}}\sin(\Delta \chi_{k,b}) \big|^2 \Big).
\end{aligned}
\end{equation*}
Note that a sufficiently large penalty factor $\rho_{\theta} \gg 1$ guarantees the equivalence between problems \eqref{eq:formul_3.2} and \eqref{eq:formul_3.1}. To handle the non-convex penalty term $\mathcal{P}_{\theta}$, we employ the MM framework with the Lipschitz constant $L_{\Theta}$ derived in Appendix B to construct a valid global upper bound~\cite{conference}:
\begin{equation} \label{eq:penalty_surrogate_theta}
\begin{aligned}
    \tilde{\mathcal{P_{\theta}}}^{(t_3)}(\boldsymbol{\theta}, \mathbf{A}) 
    &= \mathcal{P_{\theta}}(\boldsymbol{\theta}^{(t_3)}, \mathbf{A}^{(t_3)})\\
    & \hspace{-2em}  + \langle \nabla_{\mathbf{A}}\mathcal{P_{\theta}}, \mathbf{A} - \mathbf{A}^{(t_3)} \rangle
     + \nabla_{\boldsymbol{\theta}}\mathcal{P_{\theta}}^T (\boldsymbol{\theta} - \boldsymbol{\theta}^{(t_3)}) \\[-0.2em]
     & \hspace{-2em} + \frac{L_{\Theta}}{2} \big( \|\mathbf{A} - \mathbf{A}^{(t_3)}\|_F^2 + \|\boldsymbol{\theta} - \boldsymbol{\theta}^{(t_3)}\|_2^2 \big).
\end{aligned}
\end{equation}
Following \eqref{eq:R_k_surrogate}, we apply MM together with a first-order Taylor approximation to problem \eqref{eq:formul_3.2}, yielding the concave surrogate ${R}_k^{(t_3)}$. The resulting convex approximation of C2 is denoted by $\widehat{\mathrm{C2}}: \frac{T - \tau_{\text{adj}}}{T}\tilde{R}_{k}^{(t_3)}(\mathbf{A}) \ge R_{\min}, \forall k$, where $\widehat{\mathrm{C2}} \Rightarrow \mathrm{C2}$. Similarly, we adopt the MM technique and a Lipschitz surrogate for the non-convex rotation energy $E_{\text{rot}}(\boldsymbol{\theta})$~as:
\begin{flalign}
\label{eq:energy_surrogate_theta}
\begin{aligned}
&\quad \tilde{E}_{\text{rot}}^{(t_3)}(\boldsymbol{\theta}) =  E_{\text{rot}}(\boldsymbol{\theta}^{(t_3)}) + \\
&\quad\nabla_{\boldsymbol{\theta}} E_{\text{rot}}(\boldsymbol{\theta}^{(t_3)})^T (\boldsymbol{\theta}-\boldsymbol{\theta}^{(t_3)}) + \frac{L_E}{2}\|\boldsymbol{\theta}-\boldsymbol{\theta}^{(t_3)}\|_2^2,
\end{aligned} &&
\end{flalign}
where the Lipschitz constant $L_E$ is derived in Appendix B, and the variables with superscript $(t_3)$ denote the corresponding solutions obtained at the $t_3$-th iteration. Hence the total energy is written as 
\begin{equation}
\label{eq:E_total_surrogate}
    \tilde{E}_{\text{total}}^{(t_3)} = \sum_{b=1}^{B} \left( E_{\text{move},b} + \tilde{E}_{\text{rot},b}^{(t_3)}(\boldsymbol{\theta}) \right) + (T - \tau_{\text{adj}}) P_{\text{tx}} + T P_s.
\end{equation}
By incorporating the Dinkelbach transform in Appendix A and the rate surrogate $\tilde{R}_k^{(t_3)}(\mathbf{A})$, the $(t_3+1)$-th penalty-based MM iteration can be represented as:
\begin{flalign}
\label{eq:formul3.4}
\begin{aligned}
    \underset{\boldsymbol{\theta}, \{\mathbf{A}_{k}\}}{\text{maximize}} \quad & (T - \tau_{\text{adj}}) \sum_{k=1}^{K} \tilde{R}_{k}^{(t_3)}(\mathbf{A}) \\
    &\quad - \rho_{\theta} \tilde{\mathcal{P}}_{\theta}^{(t_3)}(\boldsymbol{\theta}, \mathbf{A}) - \lambda^{(m)} \tilde{E}_{\text{total}}^{(t_3)} \\
    \text{s.t.} \quad &  \widehat{\text{C2}}, \widehat{\text{C4}}, \text{C6}, \text{C14}, \widehat{\overline{\text{C15}}}.
\end{aligned} &&
\end{flalign}
where $\widehat{\overline{\mathrm{C15}}} \Rightarrow \overline{\mathrm{C15}}$, which can be derived similarly to~\eqref{eq:C11_linear}, and the details are omitted for brevity. Therefore, problem~\eqref{eq:formul3.4} is convex and can be directly solved by CVX.

\subsection{Block 4: Optimization of Time Allocation}
With $\mathbf{W}$, $\mathbf{q}$, and $\boldsymbol{\theta}$ being fixed, the time-allocation variable $\tau_{\mathrm{adj}}$ only appears in constraints C2--C4. Define $\tau_{\min} = \max_{b} \left\{ \frac{\|\mathbf{q}_b - \mathbf{q}_b^{(0)}\|_2}{v_b}, \frac{\Delta\phi_b}{\omega_b} \right\}$ and $\tau_{\max} = \min_k \left\{ T \left( 1 - \frac{R_{\min}}{R_k} \right) \right\}$, which yields the feasible interval $\tau_{\mathrm{adj}}\in[\tau_{\min},\tau_{\max}]$. Then, the resulting subproblem is given by
\begin{equation} \label{eq:tau_subproblem}
\begin{aligned}
   \hspace{-1em} \max_{\tau_{\min} \le \tau_{\text{adj}} \le \tau_{\max}} ~ & J(\tau_{\text{adj}}) \!=\! (T \!-\! \tau_{\text{adj}}) R_{\text{sum}} - \lambda^{(m)} E_{\text{total}}(\tau_{\text{adj}}) \\
    \text{s.t.} ~ & \text{C2}, \text{C3}, \text{C4}, \text{C7}.
\end{aligned}
\end{equation}

Since $E_{\mathrm{mech}}$ is independent of $\tau_{\mathrm{adj}}$, the objective $J(\tau_{\mathrm{adj}})$ is affine in $\tau_{\mathrm{adj}}$ with derivative $\frac{\partial J}{\partial \tau_{\text{adj}}} = \lambda^{(m)} P_{\text{tx}} - R_{\text{sum}}$.
Hence, the optimum is attained at either $\tau_{\min}$ or $\tau_{\max}$, i.e.,
\begin{equation}
    \tau_{\text{adj}}^{\text{opt}} =
    \begin{cases}
    \tau_{\min}, & \text{if } R_{\text{sum}} > \lambda^{(m)} P_{\text{tx}}, \\
    \tau_{\max}, & \text{otherwise}.
    \end{cases}
    \label{eq:opt_tau}
\end{equation}

From a physical perspective, the optimal time allocation reflects the trade-off between throughput and energy consumption. Specifically, the system selects $\tau_{\min}$ when the achievable rate gain exceeds the corresponding transmit-power cost $\lambda^{(m)}P_{\mathrm{tx}}$; otherwise, it chooses $\tau_{\max}$ to reduce energy consumption while still satisfying the users' QoS requirements.
\begin{figure}[!t]
    \centering
  
    \includegraphics[width=0.8\columnwidth]{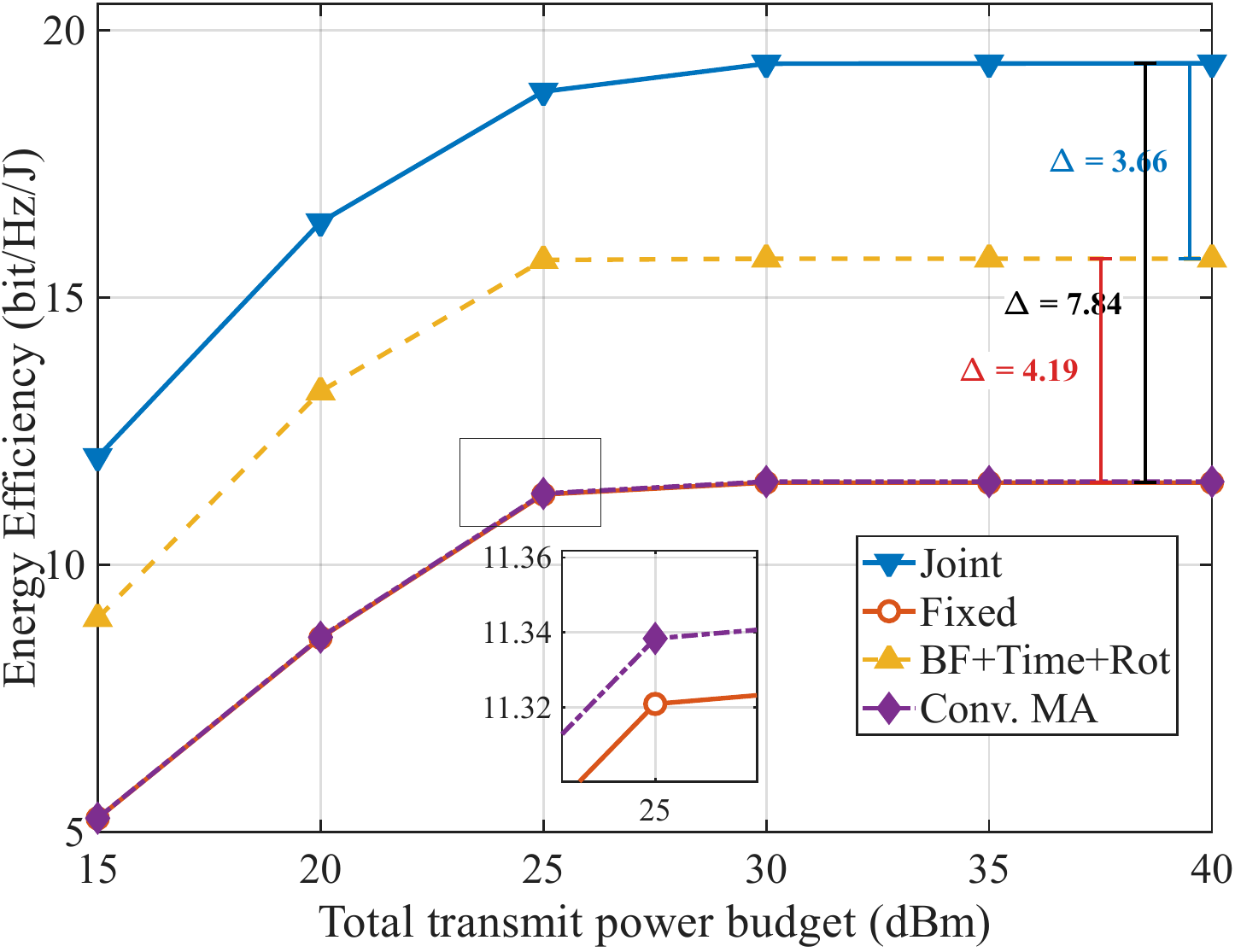} 

    \captionsetup{singlelinecheck=off, justification=raggedright}
    
    \caption{ EE versus $P_{\max}$ of different antenna systems.}
    \label{fig:ee_pmax}
  
\end{figure}
Moreover, the proposed BCD algorithm is guaranteed to converge to a suboptimal solution of~\eqref{prob:P1} with polynomial computational complexity~\cite{conference}. Specifically, the algorithm follows a four-block Dinkelbach-assisted BCD framework. The outer loop employs Dinkelbach’s transform to handle the fractional EE objective, while the inner loop successively updates the beamforming variables, antenna positions, antenna orientations, and time allocation via SDR- and MM-based optimization. Since each block update is designed to yield a non-decreasing EE and the EE is upper-bounded over the feasible set, the objective sequence generated by the algorithm converges monotonically. Therefore, the proposed algorithm monotonically converges to a suboptimal solution of~\eqref{prob:P1}.

\section{Simulation Results}
We evaluate the proposed 6DMA system at $f_c=6$ GHz. Unless otherwise specified, the system parameters are set as $B=4$, $N_s=4$, $K=4$, $\eta=2$, $P_s=1$ W, $R_{\min}=1$ bps/Hz, $T=1$ s, $v_{\max}=0.5$ m/s, $\omega_{\max}=2\pi$ rad/s, $\mathcal{V}=[-0.5,0.5]^3~\mathrm{m}^3$, $\alpha,\gamma\in[-\pi,\pi]$, $\beta\in(-\pi/2,\pi/2)$, $D_{\min}=0.10$ m, $l_0=5\times10^{-3}$ m, $M_{\mathrm{step}}=2.5\times10^{-3}$ N$\cdot$m, and $M_{\mathrm{rot}}=1\times10^{-2}$ N$\cdot$m~\cite{ding_energy_2025,wei_energy-efficient_2025}. Each movable surface is modeled as a $2\times2$ uniform planar array (UPA) with inter-element spacing $d_e=\lambda/2$, and the effective receiver noise power is set to $\sigma^2=-80$ dBm. The BS reference point is located at $(0,0,10)$ m above the user plane. The initial centers of the four antenna surfaces, defined relative to this reference point, are placed along the $x$-axis as $\mathbf{q}_b^{(0)}=[-0.45+0.30(b-1),0,0]^T$ m, which yields an inter-surface spacing of $0.30$ m. The initial orientation of all surfaces is set to $\boldsymbol{\theta}_b^{(0)}=[\pi,0,0]^T$, corresponding to a common downward boresight. The four users are located at $(80,20,0)$, $(80,-20,0)$, $(-80,20,0)$, and $(-80,-20,0)$ m. Accordingly, as seen from the BS reference point, the users lie approximately at azimuth angles $\pm14^\circ$ and $\pm166^\circ$. We adopt a LoS-dominated far-field channel model with a reference path loss of $48$ dB at 1 m and a path-loss exponent of $2$. The users are assumed to remain quasi-static within each optimization block. We compare the proposed design with three benchmarks: 1) \emph{Fixed}, which optimizes only the beamforming while keeping the antenna positions and orientations fixed; 2) \emph{BF+Time+Rot}, which jointly optimizes all variables except the antenna positions; and 3) \emph{Conv.~MA}, which jointly optimizes all variables except the 3D orientations.

\begin{figure}[!t]
    \centering
    \includegraphics[width=0.78\columnwidth]{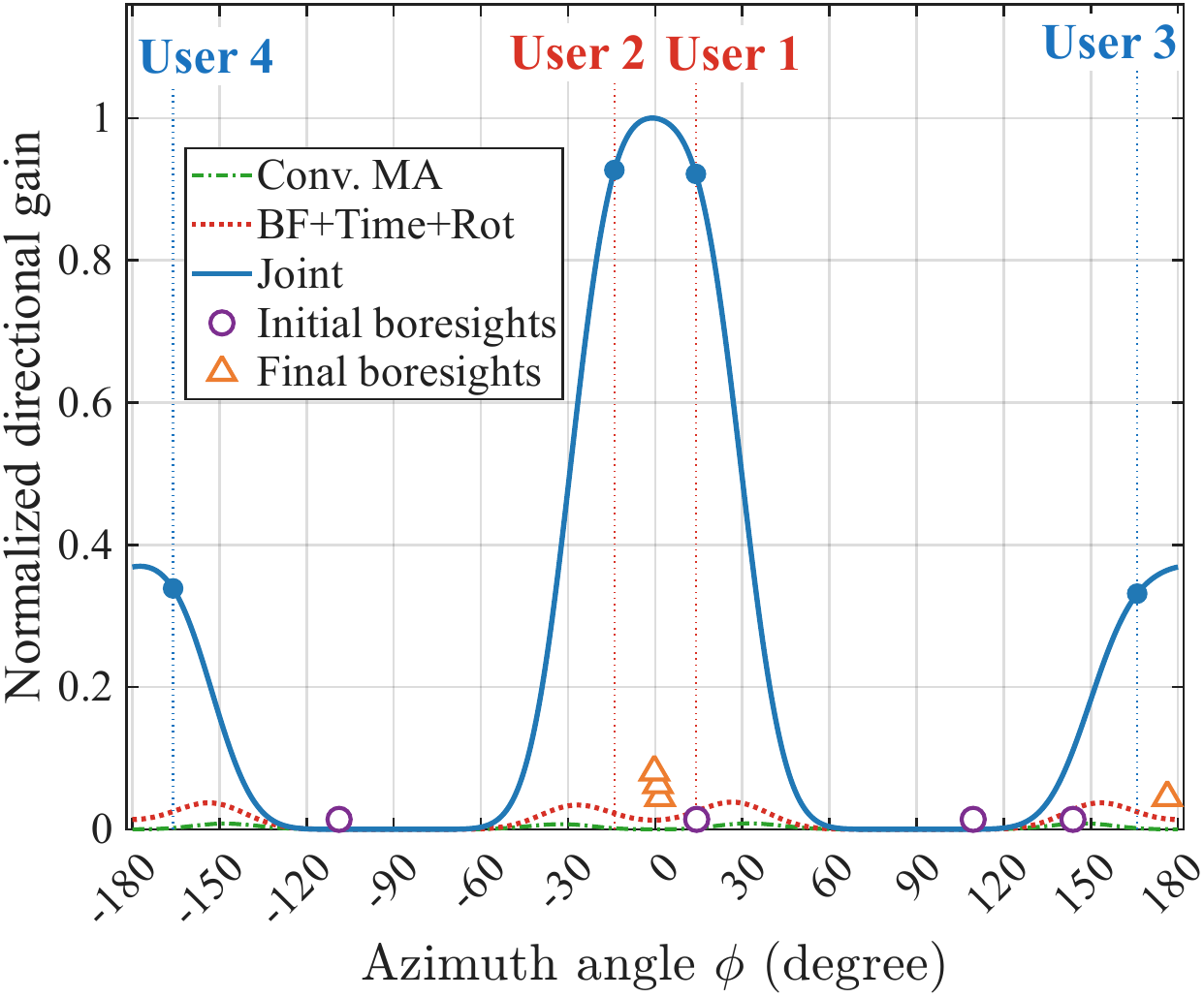} 
    \captionsetup{singlelinecheck=off, justification=raggedright}
\caption{Normalized azimuth directional-gain patterns at $P_{\max}=30$ dBm for a representative deployment. Azimuth is measured from the $+x$-axis.}
    \label{fig:gain_pattern}
    \vspace{-2mm}
\end{figure}
Fig.~\ref{fig:ee_pmax} shows the EE versus the maximum transmit-power budget, $P_{\max}$. In the low-power regime, the EE increases with $P_{\max}$ because a larger power budget enlarges the feasible beamforming set and improves the achievable sum rate. In the high-power regime, the EE gradually saturates, since the achievable rate increases only logarithmically with transmit power, whereas the total energy consumption remains non-negligible.  Furthermore, the proposed scheme consistently achieves the highest EE over the entire considered power range. This gain arises from the joint optimization of antenna positions and orientations, which offers the greatest flexibility for improving the effective channels under practical mechanical energy constraints. Moreover, \emph{BF+Time+Rot} outperforms \emph{Conv.~MA}, indicating that, in the considered LoS-dominated far-field scenario, orientation optimization is more beneficial to EE than position-only adjustment. This observation is consistent with~\eqref{eq:antenna_gain}: antenna rotation improves boresight alignment and thus the orientation-dependent element gain, whereas position optimization mainly refines inter-surface phase alignment. By jointly exploiting both mechanisms, the proposed scheme achieves the largest EE improvement, attaining a gain of $\Delta = 7.84$ bit/Hz/J over the fixed baseline.

Fig.~\ref{fig:gain_pattern} plots the normalized azimuth transmit gain at $P_{\max}=30$ dBm for the same deployment as in Fig.~\ref{fig:ee_pmax}. Compared with \emph{Conv.~MA}, the rotation-enabled schemes concentrate more radiated power toward the dominant user cluster, resulting in higher directional gain at the corresponding azimuth angles. The proposed joint design achieves the strongest directional focusing, which indicates that position optimization and orientation optimization play complementary roles in shaping the transmit pattern. In the considered LoS-dominated far-field setting, orientation optimization mainly enhances boresight alignment and hence the orientation-dependent element gain, while position optimization primarily improves coherent combining through inter-surface phase alignment. By jointly exploiting both effects, the proposed design strengthens the effective channels in the preferred directions, which helps explain its superior EE performance in Fig.~\ref{fig:ee_pmax}.

\section{Conclusion}

In this paper, we investigated the EE maximization problem for a multiuser 6DMA-enabled wireless network, taking into account a practical mechanical energy consumption model. By jointly optimizing the transmit beamforming, time allocation, and the 3D positions and orientations of the 6DMA antennas, we developed a BCD-based algorithm that integrates Dinkelbach’s transformation, SDR, the MM method, and a penalty framework to obtain a high-quality suboptimal solution. Numerical results demonstrated that the proposed design achieves significant EE gains over conventional fixed-position and mobility-constrained MA benchmarks. The results also revealed a fundamental trade-off between throughput enhancement and mechanical overhead, thus highlighting the importance of explicitly accounting for mechanical energy consumption in the practical design of the 6DMA system.

\appendices
\section*{Appendix A: Dinkelbach's Method}

To maximize the EE, we employ Dinkelbach's algorithm. Consider a generic convex fractional program:
\begin{equation}
\label{1}
    \max_{x \in \mathcal{X}} \frac{\mathcal{U}(x)}{\mathcal{E}(x)},
\end{equation}
where $\mathcal{U}(x)$ is a concave utility function, $\mathcal{E}(x)$ is a convex energy-consumption function, and $\mathcal{X}$ denotes the feasible convex set. This fractional problem can be transformed into a sequence of subtractive subproblems parameterized by $\lambda$:
\begin{equation} \label{eq:dinkelbach_obj}
    \max_{x \in \mathcal{X}} \quad \mathcal{U}(x) - \lambda \mathcal{E}(x).
\end{equation}
Given the solution $x^{(m)}$ to~\eqref{eq:dinkelbach_obj}, the update
\begin{equation}
    \lambda^{(m+1)} = \frac{\mathcal{U}(x^{(m)})}{\mathcal{E}(x^{(m)})}.
\end{equation}
guarantees convergence to the optimal solution of \eqref{1}.

\section*{Appendix B: Lipschitz Constants}
\label{app:lipschitz_derivation}

\subsection*{\underline{Lipschitz Constant $L_P$ for Position Variables (Block 2):}}

With $s \triangleq \frac{2\pi}{\lambda} \mathbf{u}_k^T (\mathbf{q}_b - \mathbf{q}_{b'})$ and $|p''(s)| \le 2$ via constraints C10 and C12, the Hessian spectral norm is governed by the Laplacian $\mathbf{L}_{\mathcal{E}}$. For a complete graph with $B$ surfaces, we have $\|\mathbf{L}_{\mathcal{E}}\|_2 = B$, summing over the $k$-th user, the resulting analytic Lipschitz constant is $L_P = 8\pi^2 K B / \lambda^2$.
\subsection*{\underline{Lipschitz Constant $L_{\Theta}$ for Orientation Variables (Block 3):}}
Let $\mathbf{a} \triangleq \mathrm{vec}_{\mathbb{R}}(\{\mathbf{A}_k\}_{k=1}^K)$ and define the residual $\mathbf{r}(\boldsymbol{\theta}, \mathbf{A}) \triangleq \mathbf{S}\mathbf{a} - \mathbf{b}(\boldsymbol{\theta})$, where $\mathbf{S}$ is the fixed selection matrix and $\mathbf{b}(\boldsymbol{\theta})$ stacks the right-hand-side terms in C16a--C17b. The penalty is $P_\theta(\boldsymbol{\theta}, \mathbf{A}) = \|\mathbf{r}(\boldsymbol{\theta}, \mathbf{A})\|_2^2$. Hence, with the stacked variable $\mathbf{x} \triangleq [\mathbf{a}^T, \boldsymbol{\theta}^T]^T$, the Hessian of $P_\theta$ with respect to $\mathbf{x}$ is
\begin{equation}
\label{eq:hessian_expansion}
    \nabla_{\mathbf{x}}^2 P_\theta(\mathbf{x}) = 2 \mathbf{J}_{\mathbf{r}}(\mathbf{x})^T \mathbf{J}_{\mathbf{r}}(\mathbf{x}) + 2 \sum_i r_i(\mathbf{x}) \nabla_{\mathbf{x}}^2 r_i(\mathbf{x}).
\end{equation}
where $\mathbf{J}_{\mathbf{r}}(\mathbf{x}) \triangleq \nabla_{\mathbf{x}} \mathbf{r}(\mathbf{x})$. Therefore, on $\mathcal{X}_\delta \triangleq \mathcal{A} \times \mathcal{T}_\delta$, the Lipschitz constant $L_{\Theta}$ can be bounded as
\begin{equation}\label{eq:analytic_Ltheta}
L_{\Theta} \triangleq \sup_{\mathbf{x} \in \mathcal{X}\delta} |\nabla{\mathbf{x}}^2 P_\theta(\mathbf{x})|2 \le 2\overline{J}^2 + 2\overline{R}\overline{H} < \infty,
\end{equation}
\begin{align*}
\label{eq:constants_def}
    \text{where} \quad \overline{J} &\triangleq \sup_{\mathbf{x} \in \mathcal{X}_\delta} \|\mathbf{J}_{\mathbf{r}}(\mathbf{x})\|_2, \quad
    \overline{R} \triangleq \sup_{\mathbf{x} \in \mathcal{X}_\delta} \|\mathbf{r}(\mathbf{x})\|_2, \\
    \text{and} \quad \overline{H} &\triangleq \sup_{\mathbf{x} \in \mathcal{X}_\delta} \left( \sum_i \|\nabla_{\mathbf{x}}^2 r_i(\mathbf{x})\|_2^2 \right)^{1/2}.
\end{align*}
Since $\mathbf{R}(\boldsymbol{\theta}_b)$ is smooth and $G_0(\phi_k,\vartheta_k,\boldsymbol{\theta}_b)=G(\mathbf{u}_k^T \mathbf{z}_b)^\eta$ on the branch $\mathbf{u}_k^T \mathbf{z}_b > 0$, $\mathbf{b}(\boldsymbol{\theta})$ is $C^2$ on $\mathcal{T}\delta$. Thus, the Lipschitz constant $L_\Theta$ is finite.

\subsection*{\underline{Lipschitz Constant $L_E$ for Rotation Energy:}}
The rotation energy $E_{\text{rot}}$ depends on the angular distance $\Delta\phi_b = \arccos\big(\frac{\text{Tr}(\mathbf{R}_{\Delta,b}) - 1}{2}\big)$. Although $\arccos(\cdot)$ exhibits singularities at $\pm 1$, physical rotation limits within $\tau_{\text{adj}}$ confine the Euler angles to a singularity-free domain. Thus, $L_E \ge \max_{\boldsymbol{\theta}} \|\nabla_{\boldsymbol{\theta}}^2 E_{\text{rot}}(\boldsymbol{\theta})\|_2$ is bounded. Backtracking line search is employed to ensure the majorization condition $\tilde{E}_{\text{rot}} \ge E_{\text{rot}}$.

\subsection*{\underline{Lipschitz Constant $L_{F,b}$ for the Trace Constraint:}}
For $\mathcal{F}_b(\boldsymbol{\theta}_b)\!\triangleq\!\operatorname{Tr}(\mathbf{R}(\boldsymbol{\theta}_b)\mathbf{C})$, with $\mathbf{C}=\mathbf{R}^T(\boldsymbol{\theta}_b^{(0)})$, the Hessian entries are
\begin{equation}
[\mathbf{H}_b]_{i,j}=\operatorname{Tr}\!\left(\frac{\partial^2 \mathbf{R}(\boldsymbol{\theta}_b)}{\partial \theta_{b,i}\partial \theta_{b,j}}\mathbf{C}\right).
\end{equation}
By Cauchy--Schwarz,
\[
|[\mathbf{H}_b]_{i,j}|
\le \Big\|\tfrac{\partial^2 \mathbf{R}}{\partial \theta_i \partial \theta_j}\Big\|_F \|\mathbf{C}\|_F
\le \sqrt{6}.
\]
Thus,
\begin{equation}\label{eq:analytic_L_Fb}
L_{F,b}=\|\mathbf{H}_b\|_2
\le \|\mathbf{H}_b\|_F
\le \sqrt{\sum_{i,j}|[\mathbf{H}_b]_{i,j}|^2}
\le 3\sqrt{6}.
\end{equation}



\begin{thebibliography}{10}
\providecommand{\url}[1]{#1}
\csname url@samestyle\endcsname
\providecommand{\newblock}{\relax}
\providecommand{\bibinfo}[2]{#2}
\providecommand{\BIBentrySTDinterwordspacing}{\spaceskip=0pt\relax}
\providecommand{\BIBentryALTinterwordstretchfactor}{4}
\providecommand{\BIBentryALTinterwordspacing}{\spaceskip=\fontdimen2\font plus
\BIBentryALTinterwordstretchfactor\fontdimen3\font minus \fontdimen4\font\relax}
\providecommand{\BIBforeignlanguage}[2]{{%
\expandafter\ifx\csname l@#1\endcsname\relax
\typeout{** WARNING: IEEEtran.bst: No hyphenation pattern has been}%
\typeout{** loaded for the language `#1'. Using the pattern for}%
\typeout{** the default language instead.}%
\else
\language=\csname l@#1\endcsname
\fi
#2}}
\providecommand{\BIBdecl}{\relax}
\BIBdecl

\bibitem{shi20256dma}
X.~Shi, X.~Shao, B.~Zheng, and R.~Zhang, ``{6DMA} {A}ided {C}ell-{F}ree {M}assive {MIMO} {C}ommunication,'' \emph{IEEE Wireless Commun. Lett.}, vol.~14, no.~5, pp. 1361--1365, May 2025.

\bibitem{zhang2026movableenergy}
Y.~Zhang, S.~Hu, D.~Mishra, and D.~W.~K. Ng, ``{Movable-Antenna Array-Enhanced Energy-Efficient {RSMA} Communication Networks},'' in \emph{Proc. IEEE ICC}, May. 2026, pp. 1--6.

\bibitem{shao_6d_2025-1}
X.~Shao, Q.~Jiang, and R.~Zhang, ``\BIBforeignlanguage{en-US}{{6D} {Movable} {Antenna} {Based} on {User} {Distribution}: {Modeling} and {Optimization}},'' \emph{\BIBforeignlanguage{en-US}{IEEE Trans. Wireless Commun.}}, vol.~24, no.~1, pp. 355--370, Jan. 2025.

\bibitem{shao_6dma_2025}
X.~Shao and R.~Zhang, ``\BIBforeignlanguage{en}{{6DMA} {Enhanced} {Wireless} {Network} with {Flexible} {Antenna} {Position} and {Rotation}: {Opportunities} and {Challenges}},'' \emph{\BIBforeignlanguage{en}{IEEE Communications Magazine}}, vol.~63, no.~4, pp. 121--128, Apr. 2025.

\bibitem{ding_energy_2025}
J.~Ding, Z.~Zhou, L.~Zhu, Y.~Zhao, B.~Jiao, and R.~Zhang, ``\BIBforeignlanguage{en-US}{Energy {Efficiency} {Maximization} for {Movable} {Antenna} {Communication} {Systems}},'' \emph{\BIBforeignlanguage{en-US}{IEEE Trans. Wireless Commun.}}, pp. 1--15, Sep. 2025.

\bibitem{wei_energy-efficient_2025}
X.~Wei, W.~Mei, X.~Huang, Z.~Chen, and B.~Ning, ``Energy-{Efficient} {Movable} {Antennas}: {Mechanical} {Power} {Modeling} and {Performance} {Optimization},'' \emph{arXiv preprint arXiv:2509.24433}, Sep. 2025.

\bibitem{spong_robot_2006}
M.~W. Spong, S.~Hutchinson, and M.~Vidyasagar, \emph{\BIBforeignlanguage{eng}{Robot {Modeling} and {Control}}}.\hskip 1em plus 0.5em minus 0.4em\relax Hoboken, N.J: Wiley, 2006.

\bibitem{murray_mathematical_1994}
R.~M. Murray, Z.~Li, S.~S. Sastry, and S.~Sastry, \emph{\BIBforeignlanguage{eng}{A {Mathematical} {Introduction} to {Robotic} {Manipulation}}}.\hskip 1em plus 0.5em minus 0.4em\relax Boca Raton, Fla London: CRC Press, 1994.

\bibitem{7913628}
X.~Yu, J.~Zhang, M.~Haenggi, and K.~B. Letaief, ``{Coverage Analysis for Millimeter Wave Networks: The Impact of Directional Antenna Arrays},'' \emph{IEEE J. Sel. Areas Commun.}, vol.~35, no.~7, pp. 1498--1512, July 2017.

\bibitem{hu_sum-rate_2025}
S.~Hu, R.~Zhao, Y.~Liao, D.~W.~K. Ng, and J.~Yuan, ``Sum-{Rate} {Maximization} for {Pinching} {Antenna}-{Assisted} {NOMA} {Systems} with {Multiple} {Dielectric} {Waveguides},'' in \emph{{IEEE} {Global} {Commun}. {Conf}. ({Globecom}) {Wkshps}.}, Dec. 2025, pp. 1--7.

\bibitem{conference}
R.~Zhao, S.~Hu, D.~Mishra, and D.~W.~K. Ng, ``{Resource Allocation for Multi-Waveguide Pinching Antenna-Assisted Broadcast Networks},'' in \emph{IEEE Global Commun. Conf. (Globecom) Wkshps.}, Dec. 2025, pp. 1--7.

\bibitem{blockage}
------, ``Robust and secure blockage-aware pinching antenna-assisted wireless communication,'' \emph{IEEE Trans. Mob. Comput.}, pp. 1--18, early access, May 22, 2026, doi: 10.1109/TMC.2026.3695952.

\bibitem{yu_irs-assisted_2021}
X.~Yu, D.~Xu, D.~W.~K. Ng, and R.~Schober, ``\BIBforeignlanguage{en-US}{{IRS}-{Assisted} {Green} {Communication} {Systems}: {Provable} {Convergence} and {Robust} {Optimization}},'' \emph{\BIBforeignlanguage{en-US}{IEEE Trans. Commun.}}, vol.~69, no.~9, pp. 6313--6329, Sep. 2021.

\bibitem{mairal_incremental_2015}
J.~Mairal, ``\BIBforeignlanguage{en-US}{Incremental {Majorization}-{Minimization} {Optimization} with {Application} to {Large}-{Scale} {Machine} {Learning}},'' \emph{\BIBforeignlanguage{en-US}{SIAM J. Optim.}}, vol.~25, no.~2, pp. 829--855, Feb. 2015.

\end{thebibliography}
\end{document}